\documentclass[sigconf]{acmart}

\usepackage{amsmath,amsthm}
\usepackage{subfig,titlesec}
\usepackage{graphicx}

\DeclareMathOperator*{\argmax}{arg\,max}
\DeclareMathOperator*{\maximize}{maximize}

\newtheorem{defn}{Definition}
\newtheorem{thrm}{Theorem}

\newtheorem{clm}{Claim}

\renewcommand\footnotetextcopyrightpermission[1]{} 
\title{On the Design of Strategic Task Recommendations for Sustainable Crowdsourcing-Based Content Moderation}

\author{Sainath Sanga}
\affiliation{
\department{Department of Computer Science \\[-0.05ex]}
\institution{Missouri University of Science and Technology \\[-0.05ex]}
\city{Rolla, Missouri}
\\[-0.05ex]
\country{United States of America
\\[-0.05ex]}
}
\email{Email: ss7db@umsystem.edu}

\author{Venkata Sriram Siddhardh Nadendla}
\affiliation{
\department{Department of Computer Science \\[-0.05ex]}
\institution{Missouri University of Science and Technology \\[-0.05ex]}
\city{Rolla, Missouri}
\\[-0.05ex]
\country{United States of America
\\[-0.05ex]}
}
\email{Email: nadendla@umsystem.edu}

\begin{abstract}
Crowdsourcing-based content moderation is a platform that hosts content moderation tasks for crowd workers to review user submissions (e.g. text, images and videos) and make decisions regarding the admissibility of the posted content, along with a gamut of other tasks such as image labeling and speech-to-text conversion. In an attempt to reduce cognitive overload at the workers and improve system efficiency, these platforms offer personalized task recommendations according to the worker's preferences. However, the current state-of-the-art recommendation systems disregard the effects on worker's mental health, especially when they are repeatedly exposed to content moderation tasks with extreme content (e.g. violent images, hate-speech). In this paper, we propose a novel, strategic recommendation system for the crowdsourcing platform that recommends jobs based on worker's mental status. Specifically, this paper models interaction between the crowdsourcing platform's recommendation system (leader) and the worker (follower) as a Bayesian Stackelberg game where the type of the follower corresponds to the worker's cognitive atrophy rate and task preferences. We discuss how rewards and costs should be designed to steer the game towards desired outcomes in terms of maximizing the platform's productivity, while simultaneously improving the working conditions of crowd workers.
\end{abstract} 


\begin{document}




\maketitle 
    
\section{Introduction}
Several users post objectionable/offensive content pertaining to sexual abuse, child pornography, groundless violence and disturbing hate-filled messages on the web. Therefore, platforms that host user-generated content (UGC) often rely on content moderation industry to eliminate any such uses. Since thousands of workers moderate numerous disturbing images/videos every day (\cite{Chen2014}, \cite{krause2016inside}, \cite{barrett2020moderates}), the recurring exposure to disturbing content can have considerable adverse effects on the worker's mental health (\cite{Chen2012}, \cite{Ghoshal2017}, \cite{barrett2020moderates}). In order to maximize the sustainability and productivity of its workforce, we propose a novel system where content moderation tasks are serviced alongside other less-strenuous tasks (e.g. animal labeling) on a crowdsourcing platform so that effective task interventions can be designed to mitigate "occupational burn-out" and sustain worker's ability to work over longer time horizons. However, crowdsourcing-based content moderation could not be scaled extensively mainly due technical challenges such as privacy concerns \cite{kittur2013future} at UGC-hosting platforms and promoting worker sustainability and productivity. This paper specifically focuses on improving the sustainability and productivity of the worker and reduce the cognitive loads on the worker.

There have been some initial attempts made to reduce the cognitive impact on the content moderators. For example, Dang \emph{et al.} designed a system in \cite{dang2018but} that blurs some parts of the inappropriate content in an image to minimize cognitive trauma. However, Karunakaran \emph{et al.} showed that blurring an image had a further negative impact on emotional affect of worker \cite{karunakaran2019testing}. Rather, simple stylistic transformations such as gray-scaling has a significant positive effect of the worker while reviewing the most violent and extreme images. On the contrary, Kaur \emph{et al.} introduced a task-independent approach in \cite{kaur2017crowdmask} for filtering potentially sensitive information using CrowdMask, a system that filters images with potentially sensitive content based on natural language definition without revealing too much information to workers along the way. However, CrowdMask works in practice only if people are recruited to create filters for sensitive content, which defeats the purpose here. Barrett \emph{et al.} presented a series of recommendations in \cite{barrett2020moderates} to improve critical aspects of content moderation and fact-checking, which are listed below: (i) End outsourcing of content moderators and raise their station in the workplace, (ii) Double the number of moderators to improve the quality of content review, (iii) Provide all moderators with top-quality, on-site medical care and (iv) Sponsor research into the health risks of content moderation.

In spite of all these attempts, the current state-of-the-art recommendation systems on a crowd sourcing platforms rely only on attributes such as worker's task ratings (\cite{ma2007effective}, \cite{xin2009social}), worker's performance and task preferences \cite{yuen2011task}, worker's cognitive abilities \cite{Isinkaye2015} and worker's stopping time \cite{monica} but still do not consider the mental state of a worker directly, thus limiting the ability to improve task performance and worker's health. The main contribution of this paper is to model and design a strategic crowdsourcing-based content moderation framework that mitigates the mental health risks at workers and promotes sustainability in crowd-workforce by designing appropriate system/worker costs and rewards. Due to the misalignment of motives at both system and worker and the lack of knowledge regarding the true mental state of a worker, we model the interaction between the task recommendation system and the worker as a Bayesian Stackelberg game. We assume that the recommender system acts as a leader and recommends a task to the worker, who then makes a decision whether to pick the system's recommendation, or choose a different task. Modeling these interactions as a Bayesian Stackelberg game helps the system to recommend appropriate jobs strategically to the worker depending on the worker's type (based on cognitive atrophy rate and the task preferences) so as to simultaneously improve system productivity as well as mitigate worker's cognitive atrophy. For example, it is beneficial for both system and worker if the worker does a content moderation job when he/she is emotionally undisturbed. On the contrary, the emotional state of a mentally disturbed worker can get worse if he/she performs a content moderation job. Therefore, we expect a desired outcome in this game when the worker does content moderation in an undisturbed state. We use the rewards and costs to make sure that their best responses steer the game towards desired outcomes. 

\section{Strategic System-Worker Interaction Model \label{Section: Model}}
\begin{figure*}[!t]
\centering
\includegraphics[width=0.60\textwidth]{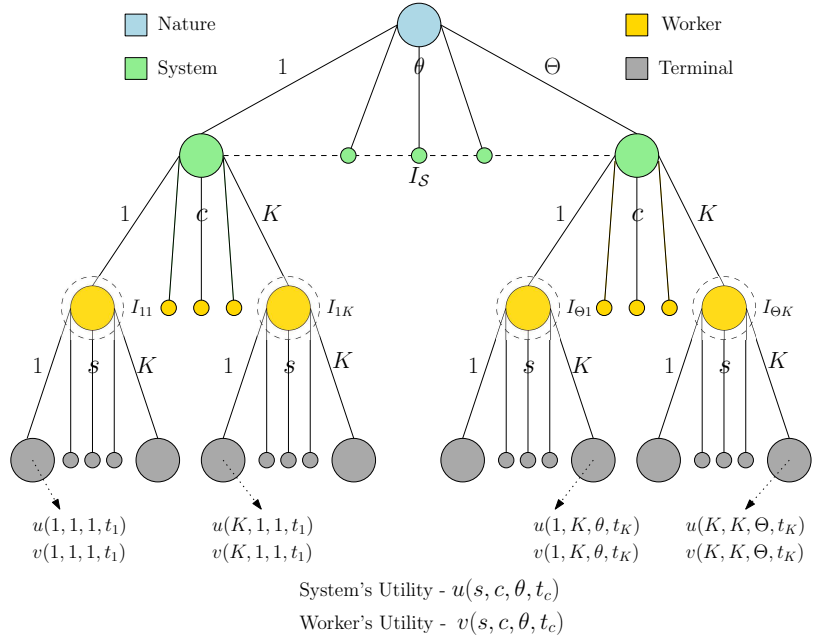}
\caption{Modeling System-Worker Interaction as a Bayesian Stackelberg Game}
\label{Fig: BSG-TJ}
\vspace{-3ex}
\end{figure*}
Consider a framework $\mathcal{N} = \{ \mathcal{S},\ \mathcal{W} \}$, where a crowd-sourcing system $\mathcal{S}$ interacts with a worker $\mathcal{W}$ as shown in Figure \ref{Fig: BSG-TJ}. Assume that any task generated on the crowdsourcing platform is categorized into the $k^{th}$ task type for some $k \in \{1, \cdots, K\}$, based on its job description, time deadline $\tau_k$ and reward $r_k$. Assume that a given $n^{th}$ worker has a linear preference order $\ell_n$ (complete and stationary) over the $K$ task types. In pursuit of several tasks in a sequential manner, the worker experiences cognitive atrophy at a rate $\beta_{n,t} \in [0,1]$ at time $t$, where $\beta_{n,t} = 1$ when the $n^{th}$ worker is completely fatigued and $\beta_{n,t} = 0$ when the $n^{th}$ worker is high-spirited. In this paper, we assume the type of the $n^{th}$ worker at time $t$ as a tuple $\boldsymbol{\theta}_{n,t} = \{ \beta_{n,t}, \ \ell_n \}$.

Note that there is little/no worker-specific data available openly to model cognitive atrophy in real content moderators, primarily due to worker's privacy concerns. However, Devaguptapu \emph{et al.} have discussed the potential fitness of discounted satisficing model in the context of content moderation. More specifically, discounted satisficing model assumes that the worker's satisficing target discounts with time at a fixed cognitive atrophy rate $\beta \in [0,1]$. However, in this paper, we assume a more general model where the cognitive atrophy rate $\beta_{n,t}$ can change with time, but falls within one of the discrete categories as shown below:
\begin{equation}
\beta_{n,t}(\beta) = 
\begin{cases}
1 & \text{ if } \beta \in [0.00-0.25], 
\\[2ex]
2 & \text{ if } \beta \in [0.26-0.50],
\\[2ex]
3 & \text{ if } \beta \in [0.51-0.75], 
\\[2ex]
4 & \text{ if } \beta \in [0.76-1]
\end{cases}
\label{Eqn: beta alpha intervals}
\end{equation}
if $\beta$ were the true rate of cognitive atrophy (e.g. discounting factor in the discounted satisficing model, as discussed in Devaguptapu \emph{et al.} in \cite{monica}) at any given moment. As a result, we can now investigate multi-round sequential system-worker interactions as a repeated game in the future. However, the scope of this paper is limited to modeling a one-shot interaction as a Bayesian Stackelberg game.


\begin{clm}
Given the cognitive atrophy rate $\beta_{n,t}$ and preference order $\ell_n$ over $K$ types of tasks of the $n^{th}$ worker, from Equation \eqref{Eqn: beta alpha intervals}, the $n^{th}$ worker  will have $\Theta = 4 \cdot (K!)$ number of possible types.
\label{Claim: Number of types}
\end{clm}
In order to reduce information overloading at the worker \cite{Isinkaye2015}, the system $\mathcal{S}$ recommends a task $s \in \{ 1, \cdots, K \}$ to the worker depending on the workers type $\theta$. The system defines a matching $m(k,\theta)$ between the $k^{th}$ task type and the worker's type $\theta$ to construct a task recommendation $s(m)$. This matching $m(k,\theta)$ is the system's judgement whether or not, a given task type $k$ is suitable to a specific worker type $\theta$, as shown below.  
\begin{defn}
Matching between any task type $k$ and any worker type $\theta$ is defined as: 
\begin{equation}
m(k, \theta) = 
\begin{cases}
0; \ \text{if task type $k$ suits worker's type $\theta$,}
\\[2ex]
1; \ \text{if task type $k$ does not suit worker's type $\theta$.}
\end{cases}
\label{Eqn: matching}
\end{equation}
\end{defn}
However, the true worker type $\theta$ is not known perfectly at the system, making it difficult to construct a task recommendation $s_m(\theta) = s$. In this paper, we assume that the system has a belief $\pi(\theta)$ regarding the worker's true type and chooses $s_m(\theta) = s$ such that its expected utility (defined in Section 3) is maximized. 

Having received a job recommendation $s$, the worker chooses a task $c\in \{ 1, \cdots, K \}$. If the worker is willing to do the recommended job $s$, they follow the given recommendation and complete the task ($c = s$). If the worker does not want to complete the recommended job $s$, they will choose another task ($c \not= s$). We assume that the worker takes time $t_c$ to complete the chosen task $c$. 

We define both the players utilities using the rewards and costs that depend on the type of the worker $\theta$ and the task $c$ that is completed by the worker. Let the rewards obtained at the system $\mathcal{S}$ and the worker $\mathcal{W}$ be denoted as $\phi_c$ and $\psi_c$, when the outcome of the interaction resulted in the completion of a job $c$. On the other hand, both the system and the worker experience a cost $\mu$ when the worker does not follow the system's recommendation $(c\not=s)$ (due to information overloading \cite{Isinkaye2015}). On the other hand, the system considers this misalignment as the cost (denoted as $\lambda$) of inefficiency in persuading the worker. 

In this paper, we assume that the worker experience a cost $\kappa_{c,\theta}$ after completing a job $c$ that does not match their type $\theta$ i.e, when $m(c,\theta) = 1$. This cost captures the impact of job $c$ on the mental state of the worker. Therefore, $\kappa_{c,\theta}$ varies depending on the type $\theta$ of the worker. For example, consider a mentally disturbed worker completing a content moderation job and a mentally undisturbed worker completing a content moderation job. The mentally disturbed worker will incur higher implicit costs, because, working on a content moderation job in a mentally disturbed state will further deteriorate the worker's mental state. We combine the rewards and costs to calculate the instantaneous utilities at both the players. Also, the utilities at both players depend on: (i) time deadline $\tau_c$, and (ii) task completion time $t_c$, for the chosen task $c$. 

\begin{defn}
The system's utilities are defined as
\begin{equation}
u(s, c,\theta, t_{c}) = 
\begin{cases}
\phi_{c}, & \text{when} \ c = s \ \text{and} \ t_{c} \leq \tau_c, 
\\[2ex]
\phi_{c}  - \lambda, & \text{when} \ c \neq s \ \text{and} \  t_c \leq \tau_c,
\\[2ex]
0, & \text{when} \ t_c > \tau_c.
\end{cases}
\label{Eqn: system utilities}
\end{equation}
\end{defn}
System attains a cost $\lambda$ when the worker chooses a task other than the recommended task ($c \not = s$). 
\begin{defn}
The worker's utilities are defined as
\begin{equation}
v(s, c,\theta, t_{c}) = 
\begin{cases}
\psi_{c} - \kappa_{c, \theta}\cdot m(c,\theta), & \text{when} \ c = s \ \text{and} \ t_{c} \leq \tau_c, 
\\[2ex]
\psi_{c} - \kappa_{c, \theta}\cdot m(c,\theta) - \mu, & \text{when} \ c \neq s \ \text{and} \  t_c \leq \tau_c,
\\[2ex]
0, & \text{when} \ t_c > \tau_c.
\end{cases}
\label{Eqn: worker utilities}
\end{equation}
\end{defn}

Given that both system and worker have non-identical utilities (i.e., mismatched motives), it is natural to model this system-worker interaction as a Stackelberg game where one player will commit to a strategy before the other players choose their own strategies \cite{fudenberg1991game}.  In our model the system acts as a leader and recommends a task $s$ to the worker. Being the follower, the worker observes the system's recommendation $s$ and choose a task $c$ depending on their type $\theta$. However, in practice, although the worker $\mathcal{W}$ may know their true type while choosing a task, the system $\mathcal{S}$ does not know the same while recommending a task. Therefore, in this paper, we assume that the system $\mathcal{S}$ has incomplete knowledge about worker's type $\theta$. Consequently, the system has a prior belief about the type of the worker $\pi(\theta)$, which produces non-trivial information sets (denoted as $I_{\mathcal{S}}$) at the system. 
In summary, since there are different types of workers, the strategic interaction between system $\mathcal{S}$ and worker $\mathcal{W}$ is modeled as a \emph{Bayesian Stackelberg game}. 

\section{Optimal Strategies \label{Section: optimal strategies}}

The System $\mathcal{S}$ recommends a task $s \in \{ 1, \cdots, K \}$ by committing to a mixed strategy $\sigma \in \Sigma$ where $\sigma(s)$ is the probability of the system recommending task $s$. Given a recommendation $s$, the worker chooses the best response that maximizes their utility $v(s, c,\theta, t_{c})$. Since the worker can observe leader’s strategy, the best response for the follower will be a pure strategy $c^*(s, \theta)$ given as: 
\begin{equation}
    c^*(s, \theta) = \argmax_{c \in \{ 1, \cdots, K \}} \: v(s, c,\theta, t_{c}). 
\label{Eqn: worker best response}
\end{equation}
Given the best response of the follower and a prior $\pi(\theta)$ over all the possible follower types $\Theta$, the expected utility of the leader's strategy $\sigma$ is  given as: 
\begin{equation}
    \mathbb{E}(u(\sigma)) = \sum_{\theta =1}^{\Theta} \pi(\theta) \sum_{s = 1}^{K} \sigma(s)\cdot u(s, c^*(s, \theta),\theta, t_{c}),
\label{Eqn: sytem expected utility}
\end{equation}
The problem of choosing an optimal strategy for the leader to commit to has been previously analyzed in \cite{conitzer2006computing, sandholm2005mixed, korzhyk2010complexity,paruchuri2008playing,letchford2009learning,letchford2010computing,paruchuri2007efficient}. One of the well-known properties about this problem is stated in the following theorem.
\begin{thrm}[Conitzer-Sandholm 2006 \cite{conitzer2006computing}]
Finding the optimal leader's strategy in a Bayesian Stackelberg game is NP-hard. 
\label{Thrm: NP-Hard}
\end{thrm}
In most cases, Bayesian Stackelberg games are analyzed according to the solution concept of a Bayes-Nash equilibrium, an extension of the Nash equilibrium for Bayesian games. However, in many settings, a Nash or Bayes-Nash equilibrium is not an appropriate solution concept, since it assumes that the agents’ strategies are chosen simultaneously \cite{conitzer2006computing}. Methods such as mixed-integer linear programs (MILPs) \cite{sandholm2005mixed}, used for finding optimal leader strategies for non-Bayesian games, can be applied to this problem by reducing the Bayesian game into a normal-form game representation using Harsanyi transformations \cite{harsanyi1972generalized}. Paruchuri \emph{et al.} have presented an efficient exact method for finding the optimal leader strategy that operates directly on the compact Bayesian representation in \cite{paruchuri2008playing}. In this method, the follower may only know the mixed strategy when choosing its strategy and it exploits the independence of the different follower types to obtain a decomposition scheme. In our game setting, the follower will know the leader's action before choosing it's own strategy i.e., the worker will know the recommendation before they choose their own job. Therefore, the method presented by Paruchuri \emph{et al.} in \cite{paruchuri2008playing} cannot be used in our game settings.

We first convert our Bayesian game to a normal-form game using Harsanyi transformations \cite{harsanyi1972generalized}. In this normal-form game, the row player is the system and the column player is the worker. The system has $K$ number of choices. Each worker can choose from $K$ tasks, therefore, there will be  $K^\Theta$ number of columns in the transformed normal form game. In the transformed normal-form game, we use the  Multiple-LP's method introduced in \cite{conitzer2006computing} to find the leader's optimal strategy using the following linear programming formulation:      

For every pure strategy $c$ of the follower,
\begin{equation}
\begin{array}{rl}
\displaystyle \maximize_{\sigma \in \Sigma} & \displaystyle \sum_{s = 1}^K \sigma(s)u(s, c,\theta, t_{c})
\\[3ex]
\text{subject to} & \text{1. } \displaystyle \sum_{s = 1}^K \sigma(s)u(s, c',\theta, t_{c}) \leq \sum_{s = 1}^K \sigma(s)u(s, c,\theta, t_{c}), 
\\[2ex]
& \qquad \qquad \qquad \qquad \qquad \qquad \forall c' \in \{ 1, \cdots, K \},
\\[-0.5ex]
& \text{2. } \displaystyle \sum_{s = 1}^K \sigma(s) = 1, \text{ and } \sigma(s) \geq 0, \ \forall s \in \{ 1, \cdots, K \}.
\end{array}
\end{equation}

The $\sigma(s)$ variables give the optimal strategy for the system. This method is polynomial in the number of columns (worker's actions) in the transformed game. For a game with $\Theta$ number of types and $K$ number of tasks, the transformed normal-form game will have $K^\Theta$ number of columns. We have to run $K^\Theta$ number of separate linear programs to solve this problem. In other words, we have the following claim.
\begin{clm}
The game with $\Theta$ number of types and $K$ number of jobs can be solved by solving $K^\Theta$ number of linear programs, each with a time complexity of $O(K^\Theta)$. In summary, the overall game complexity is of the order $O(K^{2\Theta})$.  
\label{Claim: Game Complexity}
\end{clm}

As stated in Claim \ref{Claim: Number of types}, the number of worker types grows rapidly as $\Theta = 4 \cdot K!$. Consequently, the complexity of this game is given by $O(K^{8\cdot K!})$, which grows very rapidly with increasing $K$. Therefore, it is necessary to categorize tasks into just two types $(K = 2)$ for the sake of tractability. In other words, the tasks are segregated into tasks based on whether or not it causes cognitive atrophy (e.g. as in the case of content moderation tasks). In spite of such a simplification, the problem of designing a task recommendation at each worker still remains computationally burdensome. This can severely impair the scalability of the proposed system across a large number of crowd workers. Therefore, in this paper, we present an open call for novel algorithms to the research community regarding solving this challenging problem of designing a scalable task recommender system that takes into account worker's cognitive atrophy. However, in the remaining paper, we focus our attention on reward structures to steer worker's best responses towards appropriately matching the types of workers and tasks.

\section{Reward Structures}
The approach of using a Bayesian Stackelberg Game is to make sure that system's task recommendation matches the worker's type ($\theta$). Focusing on content moderation, a desired outcome of this game will be, system recommending a content moderation task to a high spirited worker($\beta = 1$) but not to a fatigued worker ($\beta = 4)$. The Stackelberg game can be steered towards desired outcomes when the worker's best response is to choose the tasks that match their type $\theta$. Consequently, we work with worker's rewards and costs for effective persuasion towards desired outcomes.   
\begin{thrm}
A worker $\mathcal{W}$ chooses a task $k$ that matches with their own type $\theta$ when 
\begin{equation}
\psi_{-s} - \psi_s - \kappa_{-s, \theta}\cdot m(-s,\theta) \ < \ \mu \ < \ \psi_{-s} - \psi_{s} + \kappa_{s, \theta}.
\end{equation}
\label{Thrm: Reward Conditions}
\end{thrm}
\begin{proof}
We identify two possible cases based on the matching between the recommended task $s$ and the type of the worker $\theta$. In the first case, the system recommends a task that matches with the type of the worker and in the second case, the system recommends a task that does not match with the type of the worker.

\noindent
\textit{Case: 1.} When the system $\mathcal{S}$ recommends a task $s$ that matches with the type of the worker $\theta$, i.e., $m(s, \theta) = 0$, the worker's utility stated in Equation \eqref{Eqn: worker utilities} reduces to 
\begin{equation}
v(s, c,\theta, t_{c}) = 
\begin{cases}
\psi_{c}, & \text{when} \ c = s \ \text{and} \ t_{c} \leq \tau_c, 
\\[2ex]
\psi_{c} - \kappa_{c, \theta}\cdot m(c,\theta) - \mu, & \text{when} \ c \not = s \ \text{and} \  t_c \leq \tau_c,
\\[2ex]
0, & \text{when} \ t_c > \tau_c.
\end{cases}
\label{Eqn: worker utilities for case 1}
\end{equation}
Since the worker will have no incentive to choose a task that does not match their own type $\theta$, his/her best response will be to choose the recommended task $s$ if $v(s, c=s, \theta, t_c) > v(s, c \not = s, \theta, t_{c})$. Upon expanding the terms using Equation \eqref{Eqn: worker utilities for case 1}, we get
\begin{equation}
\psi_s > \psi_{-s} - \kappa_{-s, \theta}\cdot m(-s,\theta) -\mu. 
\label{Eqn: Case 1 condition}
\end{equation}
Rearranging the terms in Equation \eqref{Eqn: Case 1 condition}, we obtain the lower bound
\begin{equation}
\mu > \psi_{-s} - \psi_s - \kappa_{-s, \theta}\cdot m(-s,\theta). 
\label{Eqn: Case 1 condition}
\end{equation}

\noindent
\textit{Case: 2.} When the system $\mathcal{S}$ recommends a task $s$ that does not match with the type of the worker $\theta$, i.e., $m(s, \theta) = 1$, the utility of the system stated Equation \eqref{Eqn: worker utilities} reduces to 
\begin{equation}
v(s, c,\theta, t_{c}) = 
\begin{cases}
\psi_{c} - \kappa_{c, \theta}, & \text{when} \ c = s \ \text{and} \ t_{c} \leq \tau_c, 
\\[2ex]
\psi_{c} - \kappa_{c, \theta}\cdot m(c,\theta) - \mu, & \text{when} \ c \neq s \ \text{and} \  t_c \leq \tau_c,
\\[2ex]
0, & \text{when} \ t_c > \tau_c.
\end{cases}
\label{Eqn: worker utilities for case 2}
\end{equation}
The worker will receive higher utility when they will not follow the system's recommendation $s$ and choose a task $-s$ that matches their type $\theta$ if $v(s, c \not = s, \theta, t_{c}) > v(s, c = s, \theta, t_{c})$. Upon expanding the terms using Equation \eqref{Eqn: worker utilities for case 2}, we get
\begin{equation}
    \psi_{-s} - \mu > \psi_{s} - \kappa_{s, \theta}
\label{Eqn: Case 2 condition}
\end{equation}
Therefore from Equations \eqref{Eqn: Case 1 condition} and \eqref{Eqn: Case 2 condition}, the worker's best response will be to choose a task that matches their own type $\theta$ when
\begin{equation}
    \psi_{-s} - \psi_s - \kappa_{-s, \theta}\cdot m(-s,\theta) < \mu < \psi_{-s} - \psi_{s} + \kappa_{s, \theta}
\end{equation}
In this paper, we assume that $\kappa_{k, \theta}$ is upper bounded by $\psi_k$, because, from \eqref{Eqn: worker utilities}, if $\kappa_{k,\theta} > \psi_k$, the worker will attain negative utility and will never choose a task $k$ and the $k^{th}$ might never be completed before the time deadline $\tau_k$.  
\end{proof}
The System $\mathcal{S}$ can improve the recommendations given to the worker based on the conditions on the costs and rewards presented in Theorem \ref{Thrm: Reward Conditions} and essentially improve the overall productivity on the platform and simultaneously improving the working conditions of crowd workers. 

\section{Conclusions and Future Work}
In this paper, we modeled the interactions between task recommendation system and crowd worker in a crowdsourcing-based content moderation platform as a Bayesian Stackelberg game where the recommendation system recommends tasks depending on type of the worker. The utilities are designed depending on the matching between task type and worker type. We presented necessary conditions on worker's rewards and costs to steer the game towards desired outcomes. In the future, we will model repeated system-worker interactions using a dynamic game setting where the system can learn worker's type via updating its belief based on revealed task choices over time. Furthermore, we will also investigate trust dynamics at the worker and analyze how he/she updates their trust on system's recommendations. This is particularly interesting in a strategic setting given that the motives of the system and the worker are misaligned with each other.

\bibliographystyle{plain}
\bibliography{references.bib}
\end{document}